\documentclass{article}
\usepackage[utf8]{inputenc}
\usepackage{amsfonts}
\usepackage{amsmath}
\usepackage{amsthm}

\usepackage{hyperref}

\usepackage{tikz}
\usetikzlibrary{arrows}

\title{Signified chromatic number of grids is at most 9}
\author{
\Large Janusz Dybizba\'{n}ski\\
\small Institute of Informatics \\[-0.8ex] 
\small Faculty of Mathematics, Physics and Informatics \\[-0.8ex]
\small University of Gda\'{n}sk \\[-0.8ex]
\small 80-308 Gda\'{n}sk, Poland\\[-0.8ex]
\small {\tt jdybiz@inf.ug.edu.pl}\\
}

\date{}
\newtheorem{theorem}{Theorem}
\newtheorem{lemma}[theorem]{Lemma}

\begin{document}
\maketitle

\begin{abstract}
A signified graph is a pair $(G, \Sigma)$ where $G$ is a graph, and $\Sigma$ is a set of edges marked with '$-$'. Other edges are marked with '$+$'. A signified coloring of the signified graph $(G, \Sigma)$ is a homomorphism into a signified graph $(H, \Delta)$. The signified chromatic number of the signified graph $(G, \Sigma)$ is the minimum order of $H$.

In this paper we show that for every 2-dimensional grid $(G, \Sigma)$ there exists homomorphism from $(G, \Sigma)$ into the signed Paley graphs $SP_9$. Hence signified chromatic number of the signified grids is at most 9. This improves upper bound on this number obtained recently by Bensmail.
\end{abstract}

\bigskip\noindent \textbf{Keywords:} signified coloring; grids; Paley graphs;\\ 
\noindent\textbf{2010 Mathematics Subject Classification:} 05C15

\section{Introduction}

In the whole paper we will use standard graph theory notations. A \textit{signified graph} is a pair $(G, \Sigma)$ where $G$ is a undirected graph with an assignment to its edges one of two signs~'$+$' and~'$-$'. $\Sigma$ is a set of edges marked with~'$-$'. For the vertex $v\in V(G)$ by $N^-(v)$ (resp. $N^+(v)$) we denote the set of neighbours of $v$ such that the edge to $v$ has assignment '$-$' (resp. '$+$'). Similarly for a set $S \subset V(G)$, we define $N^{-}(S) = \bigcup_{v\in S} N^{-}(v)$ and $N^{+}(S) = \bigcup_{v\in S} N^{+}(v)$.

A \textit{signified coloring} of a signified graph $(G, \Sigma)$ is a proper coloring $\phi$ of $V(G)$ such that if there exist two edges $\{u,v\}$ and $\{x,y\}$ with $\phi(u) = \phi(x)$ and $\phi(v) = \phi(y)$, then these two edges have the same sign. The \textit{signified chromatic number} of the signified graph $(G, \Sigma)$, denoted by $\chi_2(G, \Sigma)$, is the minimum number of colors needed for a signified coloring.

Equivalently, the signified chromatic number $\chi_2(G, \Sigma)$ of the signified graph $(G, \Sigma)$ is the minimum order of the graph $(H, \Lambda)$ such that $(G, \Sigma)$ admits a signified homomorphism to $(H, \Lambda)$. Graph $(H, \Lambda)$ we call a \textit{target graph} or \textit{coloring graph}. The signified chromatic number $\chi_2(G)$ of a graph $G$ is defined as $\chi_2(G) = \max\{\chi_2(G, \Sigma) : \Sigma \subset E(G)\}$. For a graph class $\mathcal{F}$, we define the signified chromatic number $\chi_2(\mathcal{F})$ as the maximum over signified chromatic number for any members of $\mathcal{F}$. 

In this paper we focus on signified chromatic number for class of 2-dimensional \textit{grids} $\mathcal{G}$. The grid is defined as the graph being the Cartesian product of two paths. The $\chi_2(\mathcal{G})$ was investigated in 2016 by Bensmail~\cite{bens712}, who showed that $7\leq \chi_2(\mathcal{G}) \leq 12$. Recently, in 2019, the same author improved both bounds by showing, that $8\leq \chi_2(\mathcal{G}) \leq 11$.

In Section 2 we define signified Paley graphs $SP_q$ and focus on $SP_9$ as a target graph in signified homomorphism. That graph was used in this context earlier, for example Montejano et al.~\cite{monte} show that there exist signified homomorphism from every signed outerplanar graph to $SP_9$. Some other application of that target graph we can find in~\cite{planar}. In Section 3 we prove:

\begin{theorem}\label{thmm}
For every signified grid $(G,\Sigma)$ there exists a signified homomorphism $h: (G,\Sigma) \to SP_9$.
\end{theorem}

\section{Graph $SP_9$}

Let $q$ be a prime power such that $q \equiv 1 \pmod 4$ and $\mathbb{F}_q$ be finite field of order $q$. The \textit{Paley graph} $P_q$ is undirected graph with vertex set $V(P_q) = \mathbb{F}_q$ and edge set $E(P_q) = \big\{\{x,y\} : y-x $ is a square in $\mathbb{F}_q\big\}$ (note that since $-1$ is a square in $\mathbb{F}_q$, if $x-y$ is a square, then $y-x$ is a square, so the graph $P_q$ is well defined). 

The \textit{signified Paley graph} $SP_q$ is signified graph $(K_q, \Sigma)$, where $K_q$ is the complete graph on vertices $\mathbb{F}_q$ and set of negative edges $\Sigma = \big\{\{x,y\} : y-x$ is not a square in $\mathbb{F}_q\big\}$. 

The following properties of $SP_q$ are well know:

\begin{lemma}\label{l1}\cite{self}
$SP_q$ is isomorphic to the signified graph constructed by reversing signs on all edges.
\end{lemma}

\begin{lemma}\label{l2}
For every square $a\in\mathbb{F}_q$ and every $b\in \mathbb{F}_q$, the function $f(x)=ax+b$ is an automorphism in $SP_q$. This means that $SP_q$ is vertex-transitive and edge-transitive. 
\end{lemma}

To define $SP_9$ we use Galois' field $GF(3^2)$. Elements of this field are the polynomials over $GF(3)$ with multiplication modulo $x^2+1$. Elements $\{0,1,2,x,2x\}$ are squares and $\{x+1, x+2, 2x+1, 2x+2 \}$ are non squares. The graph $P_9$ is presented on Figure~\ref{fig:G}. 

\begin{figure}
    \centering
    \begin{tikzpicture}
    \tikzset{vertex/.style = {shape=circle,inner sep=1pt,fill,minimum size=.5em}}
    \tikzset{edge/.style = {_-_,> = stealth',thick}}
    \node[vertex, label=90:{$0$}] at (90:2cm) (u0) {};
    \node[vertex, label=130:{$2$}] at (130:2cm) (u1) {};
    \node[vertex, label=170:{$2x+2$}] at (170:2cm) (u2) {};
    \node[vertex, label=210:{$2x+1$}] at (210:2cm) (u3) {};
    \node[vertex, label=250:{$2x$}] at (250:2cm) (u4) {};
    \node[vertex, label=290:{$x$}] at (290:2cm) (u5) {};
    \node[vertex, label=330:{$x+2$}] at (330:2cm) (u6) {};
    \node[vertex, label=370:{$x+1$}] at (370:2cm) (u7) {};
    \node[vertex, label=410:{$1$}] at (410:2cm) (u8) {};
    \foreach \i in {0, 1, 2, ..., 7} 
        \draw[edge] (u\i) -- (u\the\numexpr\i+1\relax);
    \draw[edge] (u0) -- (u8);  
    \draw[edge] (u0) -- (u4);
    \draw[edge] (u0) -- (u5);
    \draw[edge] (u1) -- (u8);
    \draw[edge] (u3) -- (u7);
    \draw[edge] (u3) -- (u8);
    \draw[edge] (u2) -- (u4);
    \draw[edge] (u6) -- (u1);
    \draw[edge] (u6) -- (u2);
    \draw[edge] (u5) -- (u7);
    \end{tikzpicture}   
    \caption{Paley graph $P_9$}
    \label{fig:G}
\end{figure}
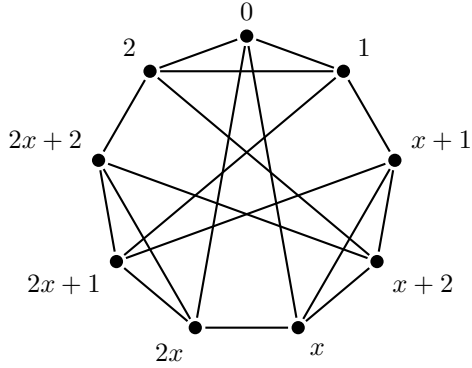

\begin{lemma}\label{l3}
For every vertex $v \in V(SP_9)$, $|N^+(v)| = 4$ and $|N^-(v)| = 4$.
\end{lemma}

We shall say that set $S \subset V(SP_9)$ is \textit{triangle free} if the subgraph of $SP_9$ induced by $S$ contains neither a triangle signed by '$+$' nor a triangle signed by '$-$'.  

\begin{lemma}\label{l4}
For every vertex $v \in V(SP_9)$, the sets $N^+(v)$ and $N^-(v)$ are triangle free.
\end{lemma}
\begin{proof}
By Lemmas~\ref{l1} and~\ref{l2}, it is enough to prove this lemma for vertex $v=0$ and $N^+(v)$. The subgraph of $SP_9$ induced by $N^+(0) = \{1,2,x,2x\}$ contains two disjoint edges signed~'$+$' ($\{1,2\}$ and $\{x,2x\}$) and the cycle $C_4$ signed~'$-$' $(1,2x,2,x)$.
\end{proof}

\begin{lemma}\label{l5}
Suppose that $S \subset V(SP_9)$, $|S|=3$ and $S$ is triangle free, then $|N^+(S)| = 8$ and $|N^-(S)| = 8$.
\end{lemma}
\begin{proof}
Since Lemma~\ref{l1}, it is enough to prove this lemma for $N^+(S)$. The graph induced by $S$ do not contain triangle signed~'$-$' so it must contain at least one edge signed~'$+$'. By Lemma~\ref{l2}, we may assume that this edge is $\{0,1\}$.
We will consider cases depending on third element of $S$:
\begin{itemize}
\setlength\itemsep{-0.3em}
    \item $S=\{0,1,2\}$ creates triangle
    \item for $S=\{0,1,x\}$, $N^+(S)=\mathbb{F}_9\setminus\{ 2x+2 \}$
    \item for $S=\{0,1,x+1\}$, $N^+(S)=\mathbb{F}_9\setminus\{ 2x+2 \}$
    \item for $S=\{0,1,x+2\}$, $N^+(S)=\mathbb{F}_9\setminus\{ x+2 \}$
    \item for $S=\{0,1,2x\}$, $N^+(S)=\mathbb{F}_9\setminus\{ x+2 \}$
    \item for $S=\{0,1,2x+1\}$, $N^+(S)=\mathbb{F}_9\setminus\{ x+2 \}$
    \item for $S=\{0,1,2x+2\}$, $N^+(S)=\mathbb{F}_9\setminus\{ 2x+2 \}$.
\end{itemize}
\end{proof}

\section{Proof of Theorem~\ref{thmm}}

Consider a path $(u,v,w)$ with arbitrary signs on the edges $\{u,v\}$ and $\{v,w\}$. Suppose that we have: an arbitrary 3-elements triangle free set $S_1 \subset V(SP_9)$ of colors available in $u$ and an arbitrary color $b \in V(SP_9)$ for the vertex $w$. Then there is a 3-elements triangle free set $S_2 \subset V(SP_9)$ available in $v$. More precisely:

\begin{lemma}\label{l6}
Consider a path $(u,v,w)$ with arbitrary signs on the edges $\{u,v\}$ and $\{v,w\}$. For every 3-elements triangle free set $S_1 \subset V(SP_9)$ and every color $b\in SP_9$, there exists a 3-elements triangle free set $S_2 \subset V(SP_9)$ such that for each $s_2 \in S_2$ there exists $s_1 \in S_1$ and coloring $c: \{u,v,w\} \to V(SP_9)$ with $c(u)=s_1$, $c(v)=s_2$, $c(w)=b$.
\end{lemma}
\begin{proof}
We will prove the lemma in case when both edges of path $(u,v,w)$ are marked with '$+$'. In any other case the proof is similar.. 

By Lemmas~\ref{l5} and~\ref{l3}, $|N^+(S_1)| = 8$ and $|N^+(b)| = 4$. Hence, there exists 3-elements set $S_2 \subset N^+(s_1) \cap N^+(b)$. By Lemma~\ref{l4}, the set $S_2$ is triangle free.
\end{proof}

\begin{proof}[Proof of Theorem~\ref{thmm}]
We color the grid $(G,\Sigma)$ row by row. It is easy to color first row by $SP_9$ (in fact, we can do it using only four colors, for example by  $N^+(0)=\{1,2,2x,2x+1\}$). Assume now that, for $k>1$, the first $k-1$ rows of $G$ have been already colored and we color $k$-th row. 

Let us denote the vertices in the $k-1$-th row by $a_1,a_2,...,a_n$ and the vertices in the $k$-th row by $b_1,b_2,...,b_n$. By Lemma~\ref{l3}, the vertex $b_1$ can be colored by four possible colors. By Lemma~\ref{l4}, any three of these colors form a triangle free set. Let us denote by $S_1$ any of these sets. Now for each $i=2,3,...,n$ we define set $S_i$ as a result of applying Lemma~\ref{l6} for the set $S_{i-1}$ and the color $h(a_i)$.

Now we can color vertices $b_1,b_2,...,b_n$ in reverse order. First we choose any color in $S_n$ for $h(b_n)$. For $h(b_{n-1})$ we set the color from $S_{n-1}$ such that the sign of the edge $(h(b_{n-1}),h(b_n))$ in $SP_9$ equals to the sign of the edge $(b_{n-1},b_{n})$ in the grid $(G,\Sigma)$. Notice that for each $s\in S_{n-1}$, the sign of the edge $(a_{n-1},b_{n-1})$ in the grid is equals to the sign of the edge $(s,h(a_{n-1}))$ in $SP_9$. Consecutive vertices we color in the same way.
\end{proof}

\end{document}